\documentclass[aps,prd,twocolumn,superscriptaddress,nofootinbib,10pt]{revtex4-2}

\usepackage{amsmath,amssymb,mathrsfs,bm}
\usepackage{graphicx}
\usepackage{xcolor}
\usepackage{amsthm}
\usepackage[colorlinks=true,linkcolor=blue,citecolor=blue,urlcolor=blue]{hyperref}
\usepackage{orcidlink}
\hypersetup{
  pdftitle={Beyond the Dolgov--Kawasaki instability: Covariant scalaron stability, entropy transport, and inner horizons in f(R) gravity},
  pdfauthor={Francisco S. N. Lobo and Francisco Tello-Ortiz}
}

\newtheorem{proposition}{Proposition}
\newtheorem{theorem}{Theorem}
\newtheorem{corollary}{Corollary}
\newtheorem{remark}{Remark}

\newcommand{\Lie}{\mathcal{L}}
\newcommand{\U}{\mathcal{U}}
\newcommand{\Msq}{m_{\rm sc}^{2}}
\newcommand{\Sx}{\mathfrak{S}_{\times}}
\newcommand{\Jx}{\mathcal{J}_{\ell|n}}
\newcommand{\Th}{\Theta}

\begin{document}

\title{\texorpdfstring{Beyond the Dolgov--Kawasaki instability:\\
Covariant scalaron stability, entropy transport, and inner horizons in $f(R)$ gravity}{Beyond the Dolgov--Kawasaki instability: Covariant scalaron stability, entropy transport, and inner horizons in f(R) gravity}}

\author{Francisco S. N. Lobo\orcidlink{0000-0002-9388-8373}}
\email{fslobo@ciencias.ulisboa.pt}
\affiliation{Instituto de Astrof\'isica e Ci\^encias do Espa\c{c}o, Faculdade de Ci\^encias da Universidade de Lisboa, Edif\'icio C8, Campo Grande, P-1749-016 Lisbon, Portugal}
\affiliation{Departamento de F\'isica, Faculdade de Ci\^encias da Universidade de Lisboa, Edif\'icio C8, Campo Grande, P-1749-016 Lisbon, Portugal}

\author{Francisco Tello-Ortiz\orcidlink{0000-0002-7104-5746}}
\email{francisco.tello@ufrontera.cl}
\affiliation{Departamento de Ciencias F\'isicas, Universidad de La Frontera, Casilla 54-D, 4811186 Temuco, Chile}

\begin{abstract}
{The Dolgov--Kawasaki instability tests the local scalaron response in metric $f(R)$ gravity, but does not determine strong-field horizon dynamics.  We separate these questions using $F=f_R$ and the Legendre potential $\mathcal U(F)=FR-f$.  The scalaron restoring coefficient is $m_{\rm sc}^2=(F-Rf_{RR})/(3f_{RR})$, whereas exact spherical Wald-area transport identities cancel the scalaron Hessian.  Their Einstein-frame interpretation distinguishes longitudinal kinetic focusing from potential-dependent cross-focusing.  We prove that, at fixed constant-curvature background with tracefree matter, transport depends only on $(f,f_R)$, while the scalaron mass depends on the independently adjustable $f_{RR}$.  Reissner--Nordstr\"om in $f(R)=R+\alpha R^2$ gives an exact discriminator: identical background transport coexists with opposite scalaron mass signs.  A charged solution with varying $F$ and a cosmological check test the identities beyond this constant-$F$ sector.  For a nonzero massive scalaron tail transmitted to a nondegenerate Cauchy horizon, a conditional fixed-background estimate gives a nonintegrable contribution to generalized entropy focusing.  This estimate does not establish nonlinear mass inflation or a particular extendibility class.  The results provide a layered framework for coupling regularity, background transport, scalaron stability, and nonlinear inner-horizon evolution.}
\end{abstract}

\maketitle

\section{Introduction}
\label{sec:intro}

Metric $f(R)$ gravity provides one of the simplest settings in which the Einstein--Hilbert action acquires an additional propagating gravitational degree of freedom.  The same simplicity that makes the theory attractive also exposes a fundamental viability question.  Dolgov and Kawasaki showed that the prototype inverse-curvature model develops an extremely rapid curvature instability in ordinary matter, despite having been designed to modify gravity only at cosmological scales~\cite{Carroll2004,NojiriOdintsov2003,DolgovKawasaki2003}.  The modern interpretation identifies the unstable mode as the scalaron, $F\equiv f_R$, and relates the high-curvature general-relativistic limit to the familiar condition $f_{RR}>0$~\cite{Faraoni2006,SawickiHu2007,Chiba2003,SotiriouFaraoni2010,DeFeliceTsujikawa2010}.  More generally, however, the local restoring coefficient is not simply the sign of $f_{RR}$; it depends on the background curvature through $(F-Rf_{RR})/(3f_{RR})$.

This observation is especially important in strong-field systems.  Cosmological fixed points, chameleon screening, relativistic stars, black-hole perturbations, and dynamical matter configurations probe different regions of the scalaron potential and impose conditions beyond the elementary DK sign test~\cite{AmendolaEtAl2007,HuSawicki2007,Starobinsky2007,Frolov2008,KobayashiMaeda2008,KobayashiMaeda2009,UpadhyeHu2009,BabichevLanglois2010,SeifertWald2007,Seifert2007,MyungMoonSon2011,ArbuzovaDolgovReverberi2013,SagunskiEtAl2018}.  A positive instantaneous scalaron mass is therefore an important local diagnostic, but not a theorem of global mode stability on an arbitrary evolving geometry.

A second, logically different, layer of strong-field physics is geometric transport.  Raychaudhuri's equation and trapping-horizon methods organize null focusing and marginal-surface evolution directly through curvature projections~\cite{Raychaudhuri1955,Hayward1994}.  In metric $f(R)$ gravity the stationary Wald entropy carries the factor $F=f_R$~\cite{Wald1993,IyerWald1994}, and recent work on dynamical black-hole entropy has sharpened the role of the generalized expansion $\Theta=\theta+k(\ln F)$ and of generalized apparent horizons~\cite{HollandsWaldZhang2024,Kong2025,JiaQiGao2026,Qi2026}.  {For Lagrangians polynomial in $R$, the increase of this entropy was established long ago by mapping the theory to its conformally related Einstein frame~\cite{JacobsonKangMyers1995}; the transport identities derived below are the local, spherically symmetric Jordan-frame counterparts of that argument.}  Complementary static analyses have likewise shown that scalaron Hessian terms modify null convergence and radial monotonicity in a way that is distinct from perturbative scalaron stability~\cite{MunozPalmaLobo2026}{, and the area-weighted double-null cross-focusing law, together with the obstructions it implies for regular inner marginal horizons, was obtained in Ref.~\cite{MunozPalmaHR2026}}.  These developments make it natural to ask whether geometric or entropy transport can diagnose the dynamical health of the scalaron.

Inner horizons make the distinction unavoidable.  Reissner--Nordstr\"om and Kerr possess Cauchy horizons in their exact stationary geometries, but perturbations are infinitely blueshifted there and generically activate the mass-inflation mechanism~\cite{PoissonIsrael1989,PoissonIsrael1990,Dafermos2003,LukOh2019}.  Double-null simulations show that mass inflation persists in charged $f(R)$ collapse even for model choices that avoid the elementary cosmological scalaron pathology~\cite{HwangLeeYeom2011}.  {In the companion analysis~\cite{LoboMunozTello2026}, we derived an exact Hawking-mass transport law and, under nondegenerate-horizon, source-dominance, and non-cancelling asymptotic assumptions, classified leading scalaron cancellation of a Price-tail source.  Cancelling its leading term is weaker than making the residual source blueshift-integrable; under the stated transverse and trace assumptions it instead forces a restricted high-curvature rigidity branch unless an independent escape channel intervenes.}  The present paper addresses the complementary, logically prior question: what can local scalaron stability itself tell us about null transport and inner-horizon stability?

Our answer is that these are distinct layers of information.  We formulate the scalaron dynamics using the Legendre potential $\mathcal U(F)=FR-f$ and derive the exact linearized trace equation on an arbitrary background.  In spherical symmetry the constant-$F$ sector reduces to a characteristic scalar wave equation on the two-dimensional normal geometry.  Independently, the generalized Wald-area variable $Fr^2$ satisfies exact null-transport identities in which the second derivatives of $F$ cancel.  The resulting \emph{potential-jet hierarchy} separates the value $\mathcal U$, the background relation $\mathcal U_{,F}=R$, and the scalaron response $\mathcal U_{,FF}=1/f_{RR}$.

The logical separation becomes exact in quadratic electrovacuum gravity.  Reissner--Nordstr\"om is a solution of $f(R)=R+\alpha R^2$ for every $\alpha$: the geometry, the null projections, the Wald-area transport, and the outer--inner horizon pair are identical, while the scalaron mass $m_{\rm sc}^2=1/(6\alpha)$ changes sign.  The same background therefore supports a non-tachyonic exterior scalaron for $\alpha>0$ and a tachyonic scalaron for $\alpha<0$.  Conversely, within one and the same geometry the transverse transport discriminator changes sign between the outer and inner horizons while the scalaron mass is constant {(its horizon value is $2\kappa_\pm/r_\pm$, proportional to the signed surface gravity and unchanged throughout this fixed-background family)}.  This provides a direct same-background proof that transport and scalaron stability cannot be substituted for one another.  {That general-relativistic solutions with vanishing or constant curvature persist in large classes of $f(R)$ theories, while their perturbations differ, is well known~\cite{Psaltis2008,BarausseSotiriou2008,delaCruzDombriz2009}.  What is new here is the use of this degeneracy as an exact discriminator between transport and stability, and its sharpening into a one-jet theorem: on every constant-curvature background with tracefree matter the transport data are fixed by $(f,f_R)$ at the background curvature and the scalaron mass by $f_{RR}$.  We also test the identities on a static charged solution with nonconstant $F$ and derive a conditional blueshift estimate for an excited massive scalaron tail, assuming nonzero transmission to a nondegenerate Cauchy horizon.}

The paper is organized into six sections.  Section~\ref{sec:scalaron} develops the scalaron dynamics, the covariant linear response equation, and its spherical and characteristic reductions.  Section~\ref{sec:transport} derives the null and Wald-area transport laws directly, {gives their Einstein-frame dictionary,} introduces the potential-jet hierarchy, and proves the transport--stability separation theorem.  Section~\ref{sec:RN} gives the exact Reissner--Nordstr\"om discriminator{, generalizes it to a one-jet theorem on constant-curvature backgrounds, and examines a static background with a running scalaron}.  Section~\ref{sec:inner} connects the result to Cauchy horizons and formulates a layered strong-field viability program.  Section~\ref{sec:conclusions} synthesizes the implications and outlines future directions.  {Appendix~\ref{app:FLRW} verifies the transport identities on a cosmological background.}

\section{Scalaron dynamics and covariant response}\label{sec:scalaron}

\subsection{Field equations and trace dynamics}

We consider four-dimensional metric $f(R)$ gravity with minimally coupled matter,
\begin{equation}
S=\frac{1}{16\pi}\int d^4x\sqrt{-g}\,f(R)+S_m,
\label{eq:action}
\end{equation}
using metric signature $(-,+,+,+)${ and geometrized units $G=c=1$ throughout.  Thus $M$, $Q$, and radii have dimensions of length, $R$, $f$, $\U$, and $m_{\rm sc}^2$ of inverse length squared, and $F$ is dimensionless.} Defining $F\equiv f_R\equiv df/dR$, the field equations can be written as
\begin{equation}
F G_{\mu\nu}
=8\pi T_{\mu\nu}
+\nabla_\mu\nabla_\nu F
-g_{\mu\nu}\Box F
+\frac12(f-FR)g_{\mu\nu}.
\label{eq:field}
\end{equation}
Their trace is
\begin{equation}
3\Box F+FR-2f=8\pi T,
\label{eq:trace}
\end{equation}
where $T=g^{\mu\nu}T_{\mu\nu}$.

Whenever $f_{RR}\neq0$ on the branch under consideration, $R$ may be regarded locally as a function of $F$.  Introduce the Legendre potential
\begin{equation}
\U(F)\equiv F R(F)-f(R(F)).
\label{eq:Udef}
\end{equation}
Then
\begin{equation}
\U_{,F}=R,
\qquad
\U_{,FF}=\frac{dR}{dF}=\frac{1}{f_{RR}},
\label{eq:Ujets}
\end{equation}
and the trace equation becomes
\begin{equation}
3\Box F+2\U-F\U_{,F}=8\pi T.
\label{eq:traceU}
\end{equation}
This representation separates the potential value $\U$ from its higher derivatives and will be central below.

We assume $F>0$ in regions where an Einstein-frame interpretation or a Wald-area interpretation is invoked.  This condition gives a positive effective gravitational coupling and ensures that the conformal map $\widetilde g_{\mu\nu}=Fg_{\mu\nu}$ is nondegenerate~\cite{SotiriouFaraoni2010,DeFeliceTsujikawa2010}.  It does not, by itself, guarantee scalaron stability.

\subsection{Covariant local scalaron mass parameter}

Consider first a constant-curvature or sufficiently slowly varying background.  Differentiating the nonlinear restoring term in Eq.~\eqref{eq:traceU} gives
\begin{equation}
\frac{d}{dF}\left(2\U-F\U_{,F}\right)
=\U_{,F}-F\U_{,FF}.
\end{equation}
The corresponding scalaron mass is therefore
\begin{equation}
{
\Msq
=\frac{F\U_{,FF}-\U_{,F}}{3}
=\frac{F-Rf_{RR}}{3f_{RR}}.}
\label{eq:mscalar}
\end{equation}
This is a scalar function of the background and is the appropriate local algebraic mass parameter.  The frequently quoted condition $f_{RR}>0$ is recovered when
\begin{equation}
|R f_{RR}|\ll F,
\label{eq:GRlike}
\end{equation}
so that
\begin{equation}
\Msq\simeq\frac{F}{3f_{RR}}.
\label{eq:DKlimit}
\end{equation}
For $F\simeq1$, Eq.~\eqref{eq:DKlimit} reduces to the standard high-curvature DK criterion~\cite{Faraoni2006,SawickiHu2007}.  Away from this regime, however, $f_{RR}>0$ alone is not sufficient: for $R>0$, one may have $f_{RR}>F/R$ and hence $\Msq<0$.  This is the familiar refinement that enters the stability of curved fixed points~\cite{AmendolaEtAl2007,SotiriouFaraoni2010,DeFeliceTsujikawa2010}.

For the original inverse-curvature model
\begin{equation}
f(R)=R-\frac{\mu^4}{R},
\label{eq:inverseR}
\end{equation}
one has
\begin{equation}
F=1+\frac{\mu^4}{R^2},
\qquad
f_{RR}=-\frac{2\mu^4}{R^3}.
\end{equation}
At positive high curvature,
\begin{equation}
\Msq
=-\frac{R^3}{6\mu^4}-\frac{R}{2}
\simeq-\frac{R^3}{6\mu^4}<0,
\label{eq:DKoriginalmass}
\end{equation}
which is the modern scalaron form of the violent instability identified by Dolgov and Kawasaki~\cite{DolgovKawasaki2003}.

\begin{remark}
{The limit $f_{RR}\rightarrow0$ must be treated separately.  A genuinely linear Lagrangian has no additional scalar degree of freedom.  A zero of $f_{RR}$ within a nonlinear theory, however, makes the local Legendre map singular and need not describe a regular decoupling limit.  In the quadratic family $|\Msq|\rightarrow\infty$ as $\alpha\rightarrow0$; the positive-mass limit and the increasingly rapid tachyonic limit must be distinguished from the exactly general-relativistic theory at $\alpha=0$.}
\end{remark}

\subsection{Covariant response on dynamical backgrounds}
\label{sec:linear}

\subsubsection{Exact linearized trace equation}

Let $(\bar g_{\mu\nu},\bar F,\bar T)$ be an arbitrary solution of Eq.~\eqref{eq:traceU}, and perturb
\begin{equation}
g_{\mu\nu}=\bar g_{\mu\nu}+h_{\mu\nu},
\qquad
F=\bar F+\delta F,
\qquad
T=\bar T+\delta T.
\end{equation}
Linearizing Eq.~\eqref{eq:traceU} gives the exact {first-order} identity
\begin{equation}
{
\left(\bar\Box-\bar m_{\rm sc}^2\right)\delta F
=\frac{8\pi}{3}\delta T
-\delta(\Box)\bar F,}
\label{eq:lineartrace}
\end{equation}
where
\begin{equation}
\bar m_{\rm sc}^2
\equiv
\frac{\bar F-\bar R\,\bar f_{RR}}{3\bar f_{RR}}.
\label{eq:backgroundmass}
\end{equation}
For $h_{\mu\nu}=\delta g_{\mu\nu}$,
\begin{equation}
\delta(\Box)\bar F
=-h^{\mu\nu}\bar\nabla_\mu\bar\nabla_\nu\bar F
-\left(\bar\nabla_\mu h^{\mu\rho}
-\frac12\bar\nabla^\rho h\right)\bar\nabla_\rho\bar F,
\label{eq:deltabox}
\end{equation}
with $h=\bar g^{\mu\nu}h_{\mu\nu}$.

Equation~\eqref{eq:lineartrace} is the covariant extension of the DK response equation to a general background.  It also makes clear why $\bar m_{\rm sc}^2$ is not, by itself, a complete global stability theorem when $\bar F$ varies: the metric perturbation couples through $\delta(\Box)\bar F$, and matter perturbations enter through $\delta T$.  General variational approaches to spherical stability address precisely this coupled problem~\cite{SeifertWald2007,Seifert2007}.  Nevertheless, $\bar m_{\rm sc}^2$ is a covariant local response coefficient and becomes the exact mass parameter of an autonomous scalaron equation on constant-$F$ backgrounds.

\begin{proposition}[Exact decoupling on constant-$F$ backgrounds]
\label{prop:constantF}
If $\bar F$ is constant on the background, then $\delta(\Box)\bar F=0$ and $\delta F$ is gauge invariant under infinitesimal diffeomorphisms.  The scalaron obeys
\begin{equation}
\left(\bar\Box-\bar m_{\rm sc}^2\right)\delta F
=\frac{8\pi}{3}\delta T.
\label{eq:decoupled}
\end{equation}
For tracefree matter perturbations, $\delta T=0$, the scalaron obeys a homogeneous Klein--Gordon equation that contains no metric perturbation.  It may still source the remaining linearized gravitational equations; the decoupling here is therefore one-way and exact at the level of the scalaron equation.
\end{proposition}

\begin{proof}
If $\bar F$ is constant, every derivative of $\bar F$ in Eq.~\eqref{eq:deltabox} vanishes.  Moreover, under $x^\mu\rightarrow x^\mu+\xi^\mu$, the scalar perturbation transforms as $\delta F\rightarrow\delta F+\Lie_\xi\bar F=\delta F$.  Equation~\eqref{eq:decoupled} then follows directly from Eq.~\eqref{eq:lineartrace}.
\end{proof}

This proposition includes electrovacuum backgrounds with constant curvature and will be exploited in Sec.~\ref{sec:RN}.  {Two remarks are in order.  A constant $\bar F$ forces, through Eq.~\eqref{eq:traceU}, a constant matter trace $\bar T$, so $\delta T$ is gauge invariant as well and Eq.~\eqref{eq:decoupled} relates gauge-invariant quantities only.  Moreover, neither the proposition nor its proof assumes spherical symmetry; this will matter for the rotating backgrounds discussed in Sec.~\ref{sec:onejet}.}

\subsubsection{Spherical characteristic reduction}

Consider a spherically symmetric background written invariantly as
\begin{equation}
\bar g_{\mu\nu}dx^\mu dx^\nu
=h_{ab}(x)dx^a dx^b+r^2(x)d\Omega^2,
\label{eq:sphericalmetric}
\end{equation}
where $h_{ab}$ is the Lorentzian metric on the two-dimensional normal space and $r$ is the areal radius.  For a constant-$F$ background, decompose
\begin{equation}
\delta F
=\sum_{\ell m}\frac{\chi_{\ell m}(x)}{r(x)}Y_{\ell m}(\vartheta,\varphi).
\label{eq:decompose}
\end{equation}
The reduction can be performed directly on the two-dimensional normal geometry.  For any scalar $\psi(x^a,\vartheta,\varphi)$ on the warped product~\eqref{eq:sphericalmetric},
\begin{equation}
\Box\psi
=D^aD_a\psi
+\frac{2}{r}D^a r\,D_a\psi
+\frac{1}{r^2}\Delta_{S^2}\psi,
\label{eq:scalarbox}
\end{equation}
where $D_a$ is the Levi--Civita derivative of $h_{ab}$ and $\Delta_{S^2}$ is the Laplacian on the unit two-sphere.  Setting
\begin{equation}
\psi=\frac{\chi(x)}{r(x)}Y_{\ell m},
\qquad
\Delta_{S^2}Y_{\ell m}=-\ell(\ell+1)Y_{\ell m},
\end{equation}
one finds, by direct differentiation,
\begin{equation}
r\left[D^aD_a\psi+\frac{2}{r}D^a r\,D_a\psi\right]
={\left[D^aD_a\chi-\frac{D^aD_a r}{r}\chi\right]Y_{\ell m}.}
\label{eq:warpedidentity}
\end{equation}
Substitution into Eq.~\eqref{eq:decoupled} gives
\begin{equation}
{
\left[D^aD_a-\mathcal V_{\ell}(x)\right]\chi_{\ell m}
=\frac{8\pi r}{3}\,\delta T_{\ell m},}
\label{eq:normalwave}
\end{equation}
with the covariant normal-space potential
\begin{equation}
{
\mathcal V_{\ell}
=\bar m_{\rm sc}^2
+\frac{D^aD_a r}{r}
+\frac{\ell(\ell+1)}{r^2}.}
\label{eq:normalpotential}
\end{equation}
Equations~\eqref{eq:scalarbox}--\eqref{eq:normalpotential} make explicit that the scalaron mass, the curvature of the spherical foliation, and the angular barrier enter as distinct contributions to the characteristic potential.

In double-null coordinates,
\begin{equation}
h_{ab}dx^a dx^b=-2e^{-2\sigma(u,v)}du\,dv,
\label{eq:doublenull}
\end{equation}
Eq.~\eqref{eq:normalwave} reads
\begin{equation}
\partial_u\partial_v\chi_{\ell m}
+\frac12e^{-2\sigma}\mathcal V_\ell\chi_{\ell m}
=-\frac{4\pi r}{3}e^{-2\sigma}\delta T_{\ell m}.
\label{eq:characteristic}
\end{equation}
This is a characteristic form of the scalaron response adapted to collapse and black-hole interiors.  For nonconstant $\bar F$, the exact metric-mixing term in Eq.~\eqref{eq:lineartrace} can be retained on the right-hand side and treated jointly with the remaining linearized gravitational constraints.

For a static background
\begin{equation}
ds^2=-A(r)dt^2+\frac{dr^2}{B(r)}+r^2d\Omega^2,
\label{eq:static}
\end{equation}
the normal-space geometry obeys
\begin{equation}
D^aD_a r=\frac{(AB)'}{2A}.
\label{eq:D2rstatic}
\end{equation}
{On a connected static region with $A>0$ and $B>0$, introducing the tortoise coordinate $dr_*/dr=(AB)^{-1/2}$ gives}
\begin{equation}
D^aD_a\chi
=\frac{1}{A}\left(-\partial_t^2+\partial_{r_*}^2\right)\chi.
\label{eq:D2chistatic}
\end{equation}
Consequently the homogeneous constant-$F$ scalaron equation assumes the Schr\"odinger form
\begin{equation}
\left[-\partial_t^2+\partial_{r_*}^2-V_\ell(r)\right]\chi_{\ell m}=0,
\label{eq:schrodinger}
\end{equation}
with
\begin{equation}
{
V_\ell(r)
=A\left[\bar m_{\rm sc}^2+\frac{\ell(\ell+1)}{r^2}\right]
+\frac{(AB)'}{2r}.}
\label{eq:staticpotential}
\end{equation}
For $A=B=N(r)$,
\begin{equation}
V_\ell
=N\left[\bar m_{\rm sc}^2+\frac{\ell(\ell+1)}{r^2}
+\frac{N'}{r}\right].
\label{eq:equalpotential}
\end{equation}
The standard black-hole scalaron stability problem is recovered as a special case~\cite{MyungMoonSon2011}.

\section{Null transport, Wald-area transport, and the potential-jet hierarchy}
\label{sec:transport}

\subsection{Radial null geometry}

Let $(\ell^\mu,n^\mu)$ be future-directed radial null vectors normalized by
\begin{equation}
\ell^2=n^2=0,
\qquad
\ell\cdot n=-1.
\label{eq:nullnorm}
\end{equation}
The expansions of the symmetry spheres are
\begin{equation}
\theta_{(\ell)}=\frac{2}{r}\ell(r),
\qquad
\theta_{(n)}=\frac{2}{r}n(r).
\label{eq:expansions}
\end{equation}
For an affinely parametrized radial null vector $k^\mu$, spherical symmetry removes shear and twist, and Raychaudhuri's equation gives
\begin{equation}
\frac{d\theta_{(k)}}{d\lambda}
=-\frac12\theta_{(k)}^2-R_{kk}.
\label{eq:ray}
\end{equation}
The null projection of Eq.~\eqref{eq:field} is
\begin{equation}
F R_{kk}=8\pi T_{kk}+\nabla_k\nabla_kF.
\label{eq:nullfield}
\end{equation}
Thus ordinary area focusing depends explicitly on the scalaron Hessian.

\subsection{Entropy expansion}

For stationary horizons in metric $f(R)$ gravity, Wald entropy is proportional to $F$ times the horizon area~\cite{Wald1993,IyerWald1994}.  Dynamical entropy away from stationarity is more subtle~\cite{HollandsWaldZhang2024}.  Recent analyses have shown that in $f(R)$ gravity the generalized expansion associated with $F$ plays a distinguished role in generalized apparent horizons and in the relation between dynamical and Wald entropy~\cite{Kong2025,JiaQiGao2026,Qi2026}.  We therefore use $Fr^2$ as the local sphere-integrated Wald-density variable and $\Theta_{(k)}$ as its normalized null derivative; no claim that $Fr^2$ is the most general nonstationary entropy functional is required for the exact transport identities below.  In spherical symmetry it is therefore natural to introduce
\begin{equation}
\mathscr A_W\equiv F r^2,
\label{eq:Waldarea}
\end{equation}
{so the sphere area in the Einstein frame is $4\pi\mathscr A_W$.  The stationary Wald entropy is $S_W=\pi\mathscr A_W$ when $\hbar=k_B=1$ as well.}  Along a radial null direction $k^\mu$, define
\begin{equation}
\Th_{(k)}
\equiv\frac{1}{Fr^2}\Lie_k(Fr^2)
=\theta_{(k)}+k(\ln F).
\label{eq:entropyexp}
\end{equation}
Combining Eqs.~\eqref{eq:ray} and~\eqref{eq:nullfield} gives the exact identity
\begin{equation}
{
\frac{d\Th_{(k)}}{d\lambda}
=-\frac12\theta_{(k)}^2
-\left(\frac{k(F)}{F}\right)^2
-\frac{8\pi}{F}T_{kk}.}
\label{eq:entropyfocus}
\end{equation}
The second derivative $k[k(F)]$ cancels exactly.  Hence, for $F>0$ and physical matter satisfying $T_{kk}\ge0$, the entropy expansion is nonincreasing.  {In the variables $(F,g_{\mu\nu})$, Eq.~\eqref{eq:entropyfocus} contains no explicit $f_{RR}$ or $\Msq$.  Written in curvature variables, its explicitly displayed scalaron kinetic term is $f_{RR}^2[k(R)]^2/F^2$, so this term alone does not distinguish the sign of $f_{RR}$.  This is an algebraic statement about the identity, not about the dependence of a background solution on the theory.  Moreover, $\Theta_{(k)}=\theta_{(k)}+f_{RR}k(R)/F$ itself depends linearly on $f_{RR}$.  In the GR-like regime~\eqref{eq:GRlike}, the sign of $f_{RR}$ decides the DK mass sign; in general the numerator $F-Rf_{RR}$ also matters.  The entropy-focusing identity is the local spherical counterpart of the conformal-frame second-law argument of Ref.~\cite{JacobsonKangMyers1995}, with the dictionary given in Sec.~\ref{sec:einsteinframe}.}

\begin{corollary}[Entropy focusing is not a scalaron-stability criterion]
\label{cor:focusnotstability}
Under $F>0$ and the matter NEC, Eq.~\eqref{eq:entropyfocus} enforces monotonic entropy focusing independently of the sign of $\Msq$.  Therefore entropy focusing cannot, by itself, diagnose the Dolgov--Kawasaki/tachyonic stability of the scalaron.
\end{corollary}

The statement concerns the exact background transport law.  If an unstable perturbation grows to nonlinear amplitude, the background itself eventually ceases to approximate the evolved solution; this does not alter the logical independence of the two criteria.

\subsection{Entropy-weighted cross transport}

For cross transport choose $n^\mu$ affine and parallel transport $\ell^\mu$ along it,
\begin{equation}
\nabla_n n^\mu=0,
\qquad
\nabla_n\ell^\mu=0.
\label{eq:crossgauge}
\end{equation}
Define the entropy-weighted outgoing flux
\begin{equation}
\Jx
\equiv Fr^2\Th_{(\ell)}
=\Lie_\ell(Fr^2).
\label{eq:Jdef}
\end{equation}
The derivative cancellation underlying the cross-transport law can be displayed explicitly.  In the gauge~\eqref{eq:crossgauge}, the spherical cross-focusing identity is
\begin{equation}
\Lie_n\theta_{(\ell)}
=-\theta_{(\ell)}\theta_{(n)}
-\frac{1}{r^2}+G_{\ell n}.
\label{eq:crossfocus}
\end{equation}
Using
\begin{equation}
\Jx=Fr^2\theta_{(\ell)}+r^2\ell(F),
\label{eq:Jexpanded}
\end{equation}
together with $\Lie_n(r^2)=r^2\theta_{(n)}$, one obtains
\begin{align}
\Lie_n\Jx
={}&-F+r^2FG_{\ell n}
+r^2\nabla_\ell\nabla_nF
\notag\\
&+r^2\theta_{(\ell)}n(F)
+r^2\theta_{(n)}\ell(F).
\label{eq:Jintermediate}
\end{align}
For any spherically symmetric scalar $F$,
\begin{equation}
\Box F
=-2\nabla_\ell\nabla_nF
-\theta_{(\ell)}n(F)
-\theta_{(n)}\ell(F),
\label{eq:boxdecomp}
\end{equation}
whereas the mixed projection of Eq.~\eqref{eq:field} gives
\begin{equation}
FG_{\ell n}
=8\pi T_{\ell n}
+\nabla_\ell\nabla_nF
+\Box F
+\frac12\U.
\label{eq:mixedcross}
\end{equation}
Substituting Eqs.~\eqref{eq:boxdecomp} and~\eqref{eq:mixedcross} into Eq.~\eqref{eq:Jintermediate} cancels every derivative of $F$ and yields the exact cross-transport identity
\begin{equation}
{
\Lie_n\Jx
=-F+r^2\left[8\pi T_{\ell n}+\frac12\U(F)\right].}
\label{eq:crossentropy}
\end{equation}
This cancellation is important: the transverse entropy transport is controlled algebraically by $F$, $\U$, the mixed matter flux, and the curvature scale of the symmetry spheres, rather than by the scalaron kinetic derivatives that appear before the field equations are combined.  Equivalently, define
\begin{equation}
{
\Sx
\equiv
\frac{F}{r^2}-8\pi T_{\ell n}-\frac12\U(F),}
\label{eq:Sxdef}
\end{equation}
so that
\begin{equation}
\Lie_n\Jx=-r^2\Sx.
\label{eq:Sxtransport}
\end{equation}

{At an outgoing entropy-marginal sphere, $\Th_{(\ell)}=0$ and hence $\Jx=0$, so $\Lie_n\Th_{(\ell)}=-\Sx/F$.  With future-directed $\ell$ and $n$, the additional condition $\Th_{(n)}<0$ defines the future designation.  Such a nondegenerate sphere is outer when $\Sx>0$ and inner when $\Sx<0$.  If $\Th_{(n)}>0$ it is instead past-designated; if $\Th_{(n)}=0$, the future/past designation requires separate treatment.  The transverse derivative classification depends algebraically on $F$, $\U$, the mixed physical stress, and $1/r^2$.}  {This is the entropy-weighted version of Hayward's outer/inner classification of marginal spheres~\cite{Hayward1994}.  The area-weighted counterpart of Eq.~\eqref{eq:crossentropy}, in which the transported quantity is $r^2\theta_{(\ell)}$ and the mixed source retains scalaron derivative terms that are compared with the same threshold $F/r^2$, was derived in Ref.~\cite{MunozPalmaHR2026}; the Wald-area weighting adopted here is what removes those derivative terms and leaves the purely algebraic source~\eqref{eq:Sxdef}.}

\subsection{\texorpdfstring{{Einstein-frame dictionary}}{Einstein-frame dictionary}}
\label{sec:einsteinframe}

{The derivative cancellations in Eqs.~\eqref{eq:entropyfocus} and~\eqref{eq:crossentropy} have a simple origin, which it is useful to make explicit.  For $F>0$ define}
\begin{equation}
{
\widetilde g_{\mu\nu}=Fg_{\mu\nu},
\qquad
\varphi=\sqrt{\frac{3}{16\pi}}\,\ln F,
\qquad
V(\varphi)=\frac{\U(F)}{16\pi F^2},}
\label{eq:EFmap}
\end{equation}
{so that the action~\eqref{eq:action} becomes that of general relativity for $\widetilde g_{\mu\nu}$ with a canonical scalar $\varphi$ of potential $V$ and matter coupled to $F^{-1}\widetilde g_{\mu\nu}$~\cite{SotiriouFaraoni2010,DeFeliceTsujikawa2010}.  The Einstein-frame areal radius is $\widetilde r^{\,2}=Fr^2=\mathscr A_W$: the Wald-area variable is the squared Einstein-frame areal radius, and $4\pi\mathscr A_W$ is the corresponding sphere area.  If $k^\mu$ is affine for $g_{\mu\nu}$, then $\widetilde k^\mu=F^{-1}k^\mu$ is affine for $\widetilde g_{\mu\nu}$, and with $(\widetilde\ell^\mu,\widetilde n^\mu)=(\ell^\mu,F^{-1}n^\mu)$ and $\widetilde T_{\mu\nu}=T_{\mu\nu}/F$ one finds}
\begin{align}
{\Th_{(k)}}
&{{}=F\,\widetilde\theta_{(\widetilde k)},
\qquad
\Jx=\Lie_\ell\big(\widetilde r^{\,2}\big),}
\notag\\
{\Sx}
&{{}=F^2\left[\frac{1}{\widetilde r^{\,2}}
-8\pi\big(\widetilde T_{\widetilde\ell\widetilde n}+V\big)\right].}
\label{eq:EFdictionary}
\end{align}
{In these variables Eq.~\eqref{eq:entropyfocus} is the Raychaudhuri equation of $\widetilde g_{\mu\nu}$ sourced by the null energy of $\varphi$,}
\begin{equation}
{
\frac{d\Th_{(k)}}{d\lambda}
=-\frac12\Th_{(k)}^2
+\Th_{(k)}\,k(\ln F)
-\frac32\big[k(\ln F)\big]^2
-\frac{8\pi}{F}T_{kk},}
\label{eq:EFray}
\end{equation}
{which is algebraically identical to Eq.~\eqref{eq:entropyfocus}: the term linear in $\Th_{(k)}$ accounts for $\lambda$ not being affine for $\widetilde g_{\mu\nu}$, and $\tfrac32[k(\ln F)]^2=8\pi[k(\varphi)]^2$.  Likewise Eq.~\eqref{eq:crossentropy} is the Einstein-frame cross-focusing equation, in which a minimally coupled spherically symmetric scalar contributes only through its potential, $T^{(\varphi)}_{\widetilde\ell\widetilde n}=V$.}

{Seen from the Einstein frame, longitudinal null focusing samples the scalar kinetic flux, while spherical cross-focusing samples the potential value $V$.  The potential cancels from the longitudinal null projection, and the scalar kinetic terms cancel from the spherical mixed projection.  Neither identity directly samples $V_{,\varphi\varphi}$.  We regard this as the cleanest way to understand why the separation holds.  The Jordan-frame formulation nevertheless carries information that the dictionary does not trivialize.  First, matter is minimally coupled only in the Jordan frame, so $T_{kk}\ge0$ in Eq.~\eqref{eq:entropyfocus} is the physical null energy condition, and $Fr^2$ is the Noether-charge entropy density of the frame in which test particles follow geodesics.  Second, the DK coefficient is not, in general, the curvature of $V$:}
\begin{align}
{V_{,\varphi\varphi}}
&{{}=\frac13\left(\U_{,FF}-\frac{3\U_{,F}}{F}+\frac{4\U}{F^2}\right)}
\notag\\
&{{}=\frac{\Msq}{F}+\frac{2\left(8\pi T-3\Box F\right)}{3F^2},}
\label{eq:EFmass}
\end{align}
{where the trace equation~\eqref{eq:traceU} was used in the second step.  They coincide, up to the conformal factor, on constant-$F$, tracefree backgrounds.  More generally, equality requires $8\pi T-3\Box F=0$, equivalently $2\U-F\U_{,F}=0$.  In other circumstances the displayed trace-dependent term distinguishes $V_{,\varphi\varphi}$ from the Jordan-frame response coefficient $\Msq$; a complete perturbation analysis must also retain matter and metric coupling.  Third, and most importantly, the separation survives the imposition of the field equations on complete solutions, which is the content of Sec.~\ref{sec:RN}.}

\subsection{The potential-jet hierarchy}

Equations~\eqref{eq:crossentropy} and~\eqref{eq:mscalar} expose a simple but important structure:
\begin{align}
\text{entropy cross transport}
&:\quad \U(F),
\label{eq:jet0}\\
\text{background curvature}
&:\quad \U_{,F}=R,
\label{eq:jet1}\\
\text{scalaron response}
&:\quad \U_{,FF}=\frac{1}{f_{RR}},
\label{eq:jet2}\\
\text{tachyonic mass}
&:\quad
3\Msq=F\U_{,FF}-\U_{,F}.
\label{eq:jetmass}
\end{align}
We organize these relations into what we call the \emph{potential-jet hierarchy}.  Null entropy transport samples the value of the scalaron Legendre potential, whereas local DK response samples its curvature in field space.  These are different data of the gravitational Lagrangian.

This separation is stronger than the statement that $F>0$ and $f_{RR}>0$ are different viability conditions.  The transverse entropy source can remain unchanged while the scalaron mass changes sign, as the exact example in Sec.~\ref{sec:RN} will demonstrate.

\subsection{Transport--stability separation}
\label{sec:separation}

We can now formulate the main conceptual result.

\begin{theorem}[Transport--stability separation]
\label{thm:separation}
In metric $f(R)$ gravity, the sign of the local scalaron mass $\Msq$ is not determined by either (i) the matter-resolved longitudinal entropy-focusing law~\eqref{eq:entropyfocus}, or (ii) the entropy-weighted transverse transport source~\eqref{eq:Sxdef}.
{Conversely, a scalaron mass sign does not fix the outer/inner derivative type: in Reissner--Nordstr\"om the same constant $\Msq$ occurs at entropy-marginal spheres with opposite signs of $\Sx$.  These comparisons concern fixed backgrounds and variation over theory space, not independence within every evolving solution.  Therefore the background transport laws cannot universally determine the scalaron mass sign or replace a dynamical stability analysis.}
\end{theorem}

\begin{proof}
Equation~\eqref{eq:entropyfocus} contains $F$, $k(F)$, and $T_{kk}$ but no {explicit} $f_{RR}$ or $m_{\rm sc}^2$, while the transverse source~\eqref{eq:Sxdef} depends on $F$, $\mathcal U$, $T_{\ell n}$, and $r$, rather than on $\mathcal U_{,FF}$.  These identities therefore contain no direct algebraic determination of the scalaron restoring coefficient.  Algebraic independence alone is not the full argument, because the derivatives of $\mathcal U$ are related once a particular theory is fixed.
{The decisive step is the same-background construction of Sec.~\ref{sec:RN}.  By Theorem~\ref{thm:onejet}, on every constant-curvature background with tracefree matter the complete transport data are fixed by the one-jet $(f,f_R)$ at the background curvature, whereas $m_{\rm sc}^2$ is fixed by $f_{RR}$, which can be prescribed independently.  The simplest instance is Theorem~\ref{thm:RNseparation}: one and the same Reissner--Nordstr\"om geometry, with identical transport data, occurs for all $f(R)=R+\alpha R^2$, whereas $m_{\rm sc}^2=1/(6\alpha)$ changes sign.  The converse direction is of a different and weaker nature, and we state it as such: $m_{\rm sc}^2$ is a property of the jet of the theory at the background curvature, whereas $\Sx$ is a quasi-local function on a given solution whose sign at a marginal sphere is its Hayward type~\cite{Hayward1994}.  In Reissner--Nordstr\"om it is proportional to the signed surface gravity, $\Sx(r_\pm)=2\kappa_\pm/r_\pm$ [Eq.~\eqref{eq:Sxkappa}], for every $\alpha\neq0$.}
Hence neither diagnostic determines the other, even after the field equations are imposed.
\end{proof}

\begin{remark}
The theorem does not state that scalaron dynamics and geometry fail to interact.  Equation~\eqref{eq:nullfield} shows exactly how a varying scalaron Hessian changes ordinary area convergence, and Eq.~\eqref{eq:lineartrace} shows how metric perturbations feed back into scalaron dynamics on nonconstant-$F$ backgrounds.  The statement is instead that background null/entropy transport and dynamical scalaron stability are different questions and neither can be substituted for the other.
\end{remark}

\begin{remark}
\label{rem:scope}
{The independence asserted in Theorem~\ref{thm:separation} is a statement over theory space at fixed background.  It does not exclude that, on a particular solution of a particular theory, the two diagnostics are correlated through the field equations; an explicit instance is given in Sec.~\ref{sec:running} [Eq.~\eqref{eq:TWPcorrelation}].}
\end{remark}

A concise way of organizing these logically distinct statements is
\begin{align}
	F>0,\quad T_{kk}\ge0
	&\quad\Longrightarrow\quad
	\frac{d\Theta_{(k)}}{d\lambda}\le0.
	\label{eq:focusoneway}
\end{align}
{This is a genuine one-way implication: for $F>0$, matter satisfying the null energy condition guarantees nonincreasing generalized expansion along the null congruence.  The separation statements are of a different kind.  They are not implications between inequalities, but statements about which datum of the Lagrangian each diagnostic samples and which questions it therefore leaves open; they are collected in Table~\ref{tab:filters}.  Neither longitudinal entropy focusing nor the sign of the transverse transport discriminator determines the sign of the local scalaron mass parameter, and scalaron stability does not determine the transport behavior.}

\begin{table*}[t]
\caption{{Layered strong-field diagnostics in metric $f(R)$ gravity.  The local layers sample different data along the background; the nonlinear layer requires the full evolution.  In layer I, the derivative law contains $[k(F)/F]^2$, but the expansion itself contains $k(F)$ linearly.  Ordinary area focusing is a distinct diagnostic.  The layers are numbered as in Sec.~\ref{sec:inner}; layer 0 is the coupling-regularity condition assumed throughout.}}
\label{tab:filters}
\begin{ruledtabular}
\begin{tabular}{llp{0.16\textwidth}p{0.27\textwidth}p{0.25\textwidth}}
{Layer} & {Diagnostic} & {Datum of $f$ sampled} & {Controls} & {Does not determine}\\
\hline
{0} & {$F>0$} & {$f_R$} & {sign of the coupling; Wald-area variable $Fr^2$} & {$\Msq$; horizon type}\\
{I} & {$\Th_{(k)}$, $\Sx$} & {$F$, $\U$, first derivatives of $F$} & {generalized entropy focusing; outer/inner derivative type} & {sign of $\Msq$}\\
{II} & {$\Msq$} & {$\U_{,F}$, $\U_{,FF}$} & {local scalaron response; exact mass at constant $F$} & {transport; rotating modes; layer III}\\
{III} & {fluxes, tails, $|\kappa_-|$} & {full field equations, global evolution} & {mass growth and extendibility in a specified regularity class} & {---}\\
\end{tabular}
\end{ruledtabular}
\end{table*}

For compact objects possessing an inner horizon, there is an additional and logically independent level of the problem. Even when the scalaron sector is locally stable and the null-transport conditions are well behaved, one must still determine whether the Cauchy horizon is stable against the nonlinear blueshift and counterstreaming mechanism responsible for mass inflation.

\section{Exact discriminator: Reissner--Nordstr\"om in quadratic gravity}
\label{sec:RN}

The separation theorem becomes especially transparent in the one-parameter family
\begin{equation}
f_\alpha(R)=R+\alpha R^2,
\label{eq:quadratic}
\end{equation}
coupled minimally to Maxwell electrodynamics.  The Maxwell stress tensor is tracefree, $T=0$.  {Take $M>0$ and $0<|Q|<M$, so that two distinct positive horizon radii exist.}  Consider the nonextremal Reissner--Nordstr\"om geometry
\begin{equation}
ds^2=-N(r)dt^2+\frac{dr^2}{N(r)}+r^2d\Omega^2,
\label{eq:RNmetric}
\end{equation}
with
\begin{equation}
N(r)=1-\frac{2M}{r}+\frac{Q^2}{r^2}
=\frac{(r-r_+)(r-r_-)}{r^2},
\label{eq:Ndef}
\end{equation}
and
\begin{equation}
r_\pm=M\pm\sqrt{M^2-Q^2},
\qquad
r_+r_-=Q^2.
\label{eq:horizons}
\end{equation}

Since $R=0$ identically,
\begin{equation}
f_\alpha(0)=0,
\qquad
F(0)=1,
\qquad
\U(0)=0,
\qquad
\nabla_\mu F=0.
\label{eq:RNdata}
\end{equation}
Equation~\eqref{eq:field} therefore reduces exactly to the Einstein--Maxwell equation for every value of $\alpha$.  Thus the same Reissner--Nordstr\"om spacetime is an exact solution throughout the entire quadratic family.  The point $\alpha=0$ is the general-relativistic limit, in which the additional scalar mode is absent; the scalaron comparison below concerns $\alpha\neq0$.

\subsection{Identical null and entropy transport}

For a normalized radial null pair, the Maxwell stress satisfies
\begin{equation}
T_{\ell\ell}=T_{nn}=0,
\qquad
8\pi T_{\ell n}=\frac{Q^2}{r^4}.
\label{eq:Maxwellproj}
\end{equation}
Since $F=1$ and $\U=0$, the longitudinal Ricci projections vanish,
\begin{equation}
R_{\ell\ell}=R_{nn}=0,
\label{eq:RNlong}
\end{equation}
and the entropy expansion coincides with the ordinary area expansion.  The transverse entropy source is
\begin{equation}
{
\Sx(r)=\frac{1}{r^2}-\frac{Q^2}{r^4}.}
\label{eq:RNSx}
\end{equation}
At the horizons,
\begin{align}
\Sx(r_+)
&=\frac{r_+-r_-}{r_+^3}>0,
\label{eq:Splus}\\
\Sx(r_-)
&=-\frac{r_+-r_-}{r_-^3}<0.
\label{eq:Sminus}
\end{align}
The outer--inner horizon pair is therefore realized by a transverse sign reversal, independently of $\alpha$.  {Equivalently, $\Sx=(rN)'/r^2$, so that at the horizons}
\begin{align}
{\Sx(r_\pm)}
&{{}=\frac{N'(r_\pm)}{r_\pm}=\frac{2\kappa_\pm}{r_\pm},}
\notag\\
{\kappa_\pm}
&{{}\equiv\frac12N'(r_\pm)=\pm\frac{r_+-r_-}{2r_\pm^2}.}
\label{eq:Sxkappa}
\end{align}
{The transverse discriminator at a Killing horizon is thus proportional to its signed surface gravity, with the same sign and the outer/inner derivative type in the sense of Hayward~\cite{Hayward1994}; $\Sx$ changes sign at $r=|Q|=\sqrt{r_+r_-}$, between the horizons.  In the extremal limit $r_+=r_-$ the two marginal spheres merge into a degenerate one with $\Sx=0$, for which the outer/inner classification is lost; the $\alpha$-independence of all transport data, and hence the separation, is unaffected.}

\subsection{Opposite scalaron stability on the same background}

For Eq.~\eqref{eq:quadratic}, $f_{RR}=2\alpha$,
and Eq.~\eqref{eq:mscalar} gives on the Reissner--Nordstr\"om background
\begin{equation}
{
\Msq=\frac{1}{6\alpha}{,}}
\label{eq:RNmass}
\end{equation}
{the familiar mass of the scalar mode of quadratic gravity~\cite{Stelle1978}.}
Thus $\alpha>0$ gives a non-tachyonic scalaron, while $\alpha<0$ gives a tachyonic scalaron.  Yet Eqs.~\eqref{eq:RNmetric}--\eqref{eq:Sminus} are unchanged.

The scalaron perturbation is exactly decoupled by Proposition~\ref{prop:constantF}.  Since Maxwell theory remains tracefree under perturbation, $\delta T=0$, and
\begin{equation}
{\left(\bar\Box-\frac{1}{6\alpha}\right)\delta F=0,}
\label{eq:RNscalarwave}
\end{equation}
The static radial potential is
\begin{equation}
V_\ell^{\rm RN}
=N\left[
\frac{1}{6\alpha}
+\frac{\ell(\ell+1)}{r^2}
+\frac{N'}{r}
\right].
\label{eq:RNpotential}
\end{equation}
For $r\ge r_+$,
\begin{equation}
N\ge0,
\qquad
N'=\frac{2(Mr-Q^2)}{r^3}>0,
\label{eq:Npositive}
\end{equation}
because $Mr_+-Q^2=r_+(r_+-r_-)/2>0$.  Hence, for $\alpha>0$,
\begin{equation}
V_\ell^{\rm RN}\ge0
\qquad (r\ge r_+),
\label{eq:Vpositive}
\end{equation}
and the standard conserved-energy argument excludes exponentially growing finite-energy exterior scalaron modes that are regular at the future horizon and satisfy the usual asymptotically flat boundary condition at infinity.  For $\alpha<0$, the potential tends to the negative constant $1/(6\alpha)$ at spatial infinity, and sufficiently long-wavelength scalaron perturbations lie in a tachyonic band.  This is the black-hole analogue of the DK instability criterion and agrees with the general scalaron analysis of $f(R)$ black holes~\cite{MyungMoonSon2011}.  {The growth rates are bounded.  Since $0\le N<1$ and $N'>0$ for $r\ge r_+$, one has $V_\ell^{\rm RN}\ge\Msq$.  For the self-adjoint radial operator $H_\ell=-d^2/dr_*^2+V_\ell^{\rm RN}$ on $L^2(dr_*)$, every negative spectral value $E=-\Omega^2$ therefore satisfies}
\begin{equation}
{
0<\Omega\le\Omega_{\max}\equiv\sqrt{-\Msq}=\frac{1}{\sqrt{6|\alpha|}},}
\label{eq:Omegamax}
\end{equation}
{The negative asymptotic potential produces negative continuous spectrum rather than requiring a discrete growing bound state.  Finite-energy wave packets with spectral support in this band exhibit tachyonic growth; rates approaching the bound are associated with long wavelengths far from the hole.  This distinguishes spectral growth from a normalizable separated eigenmode.}

\begin{theorem}[Same-background opposite-stability pair]
\label{thm:RNseparation}
The nonextremal Reissner--Nordstr\"om geometry is an exact background solution of the Maxwell theory coupled to $f(R)=R+\alpha R^2$ for every $\alpha$.  Its null convergence, Wald-area transport, and outer--inner transverse sign reversal are independent of $\alpha$.  Nevertheless, the scalaron is non-tachyonic and exterior-mode stable for $\alpha>0$, while it is tachyonic for $\alpha<0$.  Therefore no criterion constructed solely from the background null or entropy transport can determine scalaron stability.
\end{theorem}

This exact pair is useful because it removes a possible ambiguity in purely local arguments: the two theories do not merely share the same values of $f$ and $F$ at one curvature; they share the same complete electrovacuum spacetime solution and the same horizon transport data.

\subsection{\texorpdfstring{{Generalization: the one-jet theorem}}{Generalization: the one-jet theorem}}
\label{sec:onejet}

{The Reissner--Nordstr\"om pair is not an isolated coincidence.  It is well known that solutions of general relativity with vanishing or constant scalar curvature remain solutions of large classes of $f(R)$ theories, and that their perturbations nevertheless distinguish those theories~\cite{Psaltis2008,BarausseSotiriou2008,delaCruzDombriz2009}.  The following statement isolates exactly which data of the Lagrangian such a background can and cannot see.}

\begin{theorem}[{One-jet rigidity of transport}]
\label{thm:onejet}
{Let $(g_{\mu\nu},\Psi)$ solve the Einstein equations with cosmological constant $R_0/4$ and rescaled coupling,}
\begin{equation}
{
G_{\mu\nu}+\frac{R_0}{4}\,g_{\mu\nu}=\frac{8\pi}{F_0}\,T_{\mu\nu}[\Psi],
\qquad F_0>0,}
\label{eq:onejetbackground}
\end{equation}
{for tracefree matter, $T=0$.  Then $(g_{\mu\nu},\Psi)$ solves Eq.~\eqref{eq:field} for every $f$ whose one-jet at $R_0$ is}
\begin{equation}
{
f(R_0)=\frac12F_0R_0,
\qquad
f_R(R_0)=F_0.}
\label{eq:onejet}
\end{equation}
{In all these theories $F=F_0$, $\U=\tfrac12F_0R_0$ and $\Th_{(k)}=\theta_{(k)}$ on the background.  For an affine null congruence, let $\sigma_{\mu\nu}$ and $\omega_{\mu\nu}$ be its shear and twist.  The general Raychaudhuri law and, in spherical symmetry only, the transverse source read}
\begin{align}
{\frac{d\theta_{(k)}}{d\lambda}}
&{{}=-\frac12\theta_{(k)}^2-\sigma_{\mu\nu}\sigma^{\mu\nu}}
\notag\\
&{{}\quad+\omega_{\mu\nu}\omega^{\mu\nu}-\frac{8\pi}{F_0}T_{kk},}
\notag\\
{\Sx}
&{{}=F_0\left(\frac{1}{r^2}-\frac{R_0}{4}\right)-8\pi T_{\ell n},}
\label{eq:onejettransport}
\end{align}
{Twist vanishes for a hypersurface-orthogonal congruence, and both shear and twist vanish for the radial spherical congruences used above.  At fixed background, all geometric null transport, and the spherical entropy transport, depend on $f$ only through the one-jet~\eqref{eq:onejet}.  The scalaron mass is instead}
\begin{equation}
{
\Msq=\frac{F_0}{3f_2}-\frac{R_0}{3},
\qquad
f_2\equiv f_{RR}(R_0),}
\label{eq:onejetmass}
\end{equation}
{where $f_2$ is an arbitrary nonzero real number, not constrained by Eq.~\eqref{eq:onejet}.  In particular both signs of $\Msq$ occur on the same background with the same transport data.}
\end{theorem}

\begin{proof}
{For $R=R_0$ the scalaron $F=F_0$ is constant and Eq.~\eqref{eq:field} reduces to $F_0G_{\mu\nu}=8\pi T_{\mu\nu}-\tfrac12\U_0\,g_{\mu\nu}$.  Its trace with $T=0$ gives $\U_0=\tfrac12F_0R_0$, which is the first condition in Eq.~\eqref{eq:onejet}, and the equation becomes Eq.~\eqref{eq:onejetbackground}.  Equations~\eqref{eq:onejettransport} follow from the general affine Raychaudhuri equation and, in spherical symmetry, Eq.~\eqref{eq:Sxdef}, with $\nabla_\mu F=0$.  Finally, $f(R)=\tfrac12F_0R_0+F_0(R-R_0)+\tfrac12f_2(R-R_0)^2$ realizes any $f_2$ at fixed one-jet, and Eq.~\eqref{eq:mscalar} gives Eq.~\eqref{eq:onejetmass}; taking $f_2\to0^{\pm}$ makes $\Msq\to\pm\infty$ for any $R_0$.}
\end{proof}

{Three consequences are worth recording.  (i) For $R_0=0$, $F_0=1$ and $f_2=2\alpha$ one recovers Theorem~\ref{thm:RNseparation}.  (ii) Neither Proposition~\ref{prop:constantF} nor the first part of Theorem~\ref{thm:onejet} assumes spherical symmetry.  The Kerr--Newman family is therefore likewise common to all theories with one-jet $(0,F_0)$ at $R=0$, and its scalaron is an exactly decoupled Klein--Gordon field of mass $\Msq=F_0/(3f_2)$.  There, however, $\Msq>0$ no longer guarantees exterior mode stability: a massive scalar on Kerr can admit superradiantly unstable bound states in suitable mass and rotation regimes~\cite{Myung2011rot,ShlapentokhRothman2014}.  This mechanism applies to the decoupled scalaron when its parameters lie in such a regime; it is not a claim of instability for every positive mass and every rotating background.  The non-tachyonic condition is thus not even a sufficient mode-stability criterion beyond spherical symmetry, which reinforces the layered structure of Sec.~\ref{sec:inner}.  (iii) For $R_0\neq0$ the theorem covers the Reissner--Nordstr\"om--(anti-)de Sitter family~\cite{delaCruzDombriz2009}.  In $f=R-2\Lambda+\alpha R^2$ one has $R_0=4\Lambda$ and $F_0=1+8\alpha\Lambda$, so the one-jet itself depends on $\alpha$ (the charge term of the metric is rescaled by $1/F_0$) while $\Msq=1/(6\alpha)$ as before; to vary the scalaron mass at strictly fixed transport one moves instead along the fixed-jet family used in the proof.  For $R_0>0$ this is the setting in which the regularity of the Cauchy horizon is controlled by the ratio of the exterior spectral gap to the inner surface gravity~\cite{Cardoso2018}; the scalaron adds to that problem a massive-field family of quasinormal modes whose gap depends on $f_2$ and not on any transport datum.}

\subsection{\texorpdfstring{{A running scalaron: static backgrounds with nonconstant $F$}}{A running scalaron: static backgrounds with nonconstant F}}
\label{sec:running}

{All backgrounds covered by Theorem~\ref{thm:onejet} have constant $F$, so the entropy weighting of Sec.~\ref{sec:transport} is not exercised on them.  To test the framework where it is, consider a static solution of the form~\eqref{eq:static} with $A=B=N$ and matter obeying $T_{kk}=0$ for radial null $k^\mu$, such as a radial electric field.  Then $R_{kk}=0$ and Eq.~\eqref{eq:nullfield} gives $F''=0$,}
\begin{equation}
{
F(r)=F_0+F_1r.}
\label{eq:linearF}
\end{equation}
{In ingoing Eddington--Finkelstein coordinates the gauge~\eqref{eq:crossgauge} is realized by $n=-\partial_r$ and $\ell=\partial_v+\tfrac12N\partial_r$, so that $\Jx=\tfrac12N(Fr^2)'$ and Eq.~\eqref{eq:Sxtransport} gives the purely kinematical expression}
\begin{equation}
{
\Sx=\frac{1}{2r^2}\frac{d}{dr}\left[N\frac{d(Fr^2)}{dr}\right],}
\label{eq:Sxstatic}
\end{equation}
{which reduces to $(rN)'/r^2$ for $F=1$.  At a nondegenerate Killing horizon $r_h$,}
\begin{align}
{\Sx(r_h)}
&{{}=\frac{\kappa_h}{r_h^2}\,(Fr^2)'\big|_{r_h}
=\frac{\kappa_h\left(2F_0+3F_1r_h\right)}{r_h},}
\notag\\
{\kappa_h}
&{{}\equiv\frac12N'(r_h).}
\label{eq:Sxhorizonstatic}
\end{align}
{Two features are worth noting.  First, the entropy-weighted type of a horizon is the sign of $\kappa_h$ times the sign of $(Fr^2)'$.  It agrees with the area-based type whenever the Wald area increases outward, but the two differ in the window $-F_0<F_1r_h<-\tfrac23F_0$ (for $F_0>0$), where $F>0$ and yet $(Fr^2)'<0$.  The future/past designation also changes, since $\Theta_{(n)}=-(Fr^2)'/(Fr^2)$: its sign must be checked independently of the outer/inner derivative sign.  Second, on a locally invertible branch with $F_1\neq0$ and $R'\neq0$, one has $f_{RR}=F'/R'=F_1/R'$, so that Eq.~\eqref{eq:mscalar} becomes}
\begin{equation}
{
\Msq(r)=\frac{(F_0+F_1r)\,R'(r)-F_1R(r)}{3F_1},}
\label{eq:mscrunning}
\end{equation}
{which involves $R'$ and therefore one more radial derivative of $N$ than the horizon datum~\eqref{eq:Sxhorizonstatic}.}

{An explicit solution of this class is the charged black hole with dynamical curvature of Tang, Wang and Papantonopoulos~\cite{TangWangPapantonopoulos2021}.  In our conventions it reads}
\begin{align}
{N(r)}
&{{}=\frac12-\frac{\mu}{r}+\frac{Q^2}{F_0r^2}-\frac{R_\Lambda}{12}\,r^2,}
\label{eq:TWPmetric}\\
{R(r)}
&{{}=R_\Lambda+\frac{1}{r^2},
\qquad
\mu\equiv-\frac{F_0}{3F_1},}
\notag\\
{f(R)}
&{{}=F_0R+2F_1\sqrt{R-R_\Lambda}-\frac12F_0R_\Lambda,}
\label{eq:TWPtheory}
\end{align}
{with $F$ given by Eq.~\eqref{eq:linearF} and $F_1<0$ for $\mu>0$; we have verified that Eq.~\eqref{eq:field} with the Maxwell source~\eqref{eq:Maxwellproj} is satisfied identically, and that the two expressions~\eqref{eq:Sxdef} and~\eqref{eq:Sxstatic} for $\Sx$ agree on it.  For $R_\Lambda=0$, take $F_0>0$, $\mu>0$, and $0<Q^2<F_0\mu^2/2$.  The geometry then has a solid-angle deficit rather than being asymptotically flat, two distinct positive horizons $r_\pm=\mu\pm\sqrt{\mu^2-2Q^2/F_0}$, and $F>0$ only for $0<r<3\mu$.  The centre is singular, so this is a local benchmark away from $r=0$, not a regular or globally viable model.  Equation~\eqref{eq:mscrunning} gives}
\begin{equation}
{
\Msq(r)=\frac{2\mu-r}{r^3}-\frac{R_\Lambda}{3}.}
\label{eq:TWPmass}
\end{equation}
{For $R_\Lambda=0$ the local coefficient $\Msq$ is positive at both horizons and throughout $0<r\le r_-$ and $r_-<r\le r_+$, since $r_\pm<2\mu$, and becomes negative for $2\mu<r<3\mu$: a spatial change of sign of $\Msq$ that no constant-$F$ background can display, while the horizon classification is the Reissner--Nordstr\"om one, $\Sx(r_+)>0>\Sx(r_-)$.  Moreover, since $(Fr^2)'=F_0\,r(2\mu-r)/\mu$,}
\begin{equation}
{
\Sx(r_h)=\frac{F_0}{\mu}\,\kappa_h\,r_h^2\,\Msq(r_h)
\qquad (R_\Lambda=0):}
\label{eq:TWPcorrelation}
\end{equation}
{on this particular solution the horizon transport discriminator is proportional to the product of the surface gravity and the local scalaron mass.  This is the kind of on-shell correlation anticipated in Remark~\ref{rem:scope}.  It is a property of the specific Lagrangian~\eqref{eq:TWPtheory} and not a general relation, as the case $R_\Lambda\neq0$ already shows: there $\Msq$ is shifted by $-R_\Lambda/3$ while $(Fr^2)'$ is unchanged, and the zeros of the two quantities no longer coincide.  A scalaron-active charged benchmark of this type is analyzed with the area-weighted cross-focusing law, including its inner-horizon obstruction, in Ref.~\cite{MunozPalmaHR2026}.}

{The zero of $(Fr^2)'$ at $r=2\mu$ also has a geometric meaning.  Here $F=F_0/3>0$ and, for nonzero charge, $N(2\mu)=Q^2/(4F_0\mu^2)>0$.  Thus this sphere is not a Killing horizon, although both $\Theta_{(\ell)}=\tfrac12N(Fr^2)'/(Fr^2)$ and $\Theta_{(n)}=-(Fr^2)'/(Fr^2)$ vanish.  It is a stationary point of the Einstein-frame areal radius, with no ordinary area marginality.  This simultaneous entropy marginality lies outside the future-designated classification requiring $\Theta_{(n)}<0$.  Finally, since $F$ varies, positivity of the local mass coefficient in part of this geometry does not establish mode stability of the coupled perturbations.}

\section{Inner horizons and strong-field viability}
\label{sec:inner}

The Reissner--Nordstr\"om example makes clear that an inner horizon raises several independent questions.

\subsection{Layer I: background horizon transport}

The first question is kinematical and quasi-local: does the exact background admit an inner marginal boundary, and which null channel permits it?  In the Reissner--Nordstr\"om family, Eqs.~\eqref{eq:Splus}--\eqref{eq:Sminus} show that the answer is a transverse reversal.  Longitudinal convergence is saturated throughout the electrovacuum geometry.  This statement is unchanged by the sign of $\alpha$.  {By Eq.~\eqref{eq:Sxkappa} the reversal is that of the signed surface gravity, and by Theorem~\ref{thm:onejet} the same holds for every $f$ sharing the one-jet $(0,1)$ at $R=0$.}

\subsection{Layer II: scalaron stability}

The second question asks whether the additional $f(R)$ degree of freedom is dynamically well behaved on that background.  In the quadratic family this is decided by Eq.~\eqref{eq:RNmass}: $\alpha>0$ passes the local tachyonic test and has a nonnegative exterior scalaron potential, whereas $\alpha<0$ fails before one even asks about the nonlinear fate of the Cauchy horizon.

For a general nonconstant-curvature background, Eq.~\eqref{eq:lineartrace} shows that the corresponding test is more involved.  The covariant coefficient $\Msq$ remains the local restoring mass, but metric perturbations and matter-trace perturbations couple to it.  A complete mode or initial-value analysis therefore requires the coupled characteristic system rather than only the sign of $f_{RR}$.  {Even on constant-$F$ backgrounds, where the scalaron decouples exactly, $\Msq>0$ does not guarantee exterior mode stability: rotating black holes can support superradiantly unstable massive bound states in suitable parameter regimes~\cite{Myung2011rot,ShlapentokhRothman2014}.  Strict positivity is not a universal necessary condition either.  The massless case must be distinguished, and on an asymptotically anti-de Sitter background a negative mass squared may be compatible with the Breitenlohner--Freedman bound and appropriate boundary conditions~\cite{BreitenlohnerFreedman1982}; the sign of a local coefficient alone therefore cannot replace the global spectral problem.}

\subsection{Layer III: Cauchy-horizon stability}

The third question is nonlinear and global: if a background inner horizon exists and the scalaron is locally non-tachyonic, does the Cauchy horizon survive generic perturbations?  In general relativity the answer for the ideal Reissner--Nordstr\"om Cauchy horizon is controlled by the blueshift and counterstreaming instability associated with mass inflation~\cite{PoissonIsrael1989,PoissonIsrael1990}.  Rigorous results show that the resulting extendibility problem depends on the regularity class and the precise matter model~\cite{Dafermos2003,LukOh2019}.

Metric $f(R)$ gravity does not collapse these questions into one.  Double-null simulations by Hwang, Lee, and Yeom found mass inflation in charged collapse even for $f(R)$ models selected to have cosmologically stable higher-curvature behavior~\cite{HwangLeeYeom2011}.  The companion local analysis further shows, under explicit nondegenerate-horizon and dominance assumptions, that a blueshift-nonintegrable effective longitudinal source produces Hawking-mass inflation, while leading scalaron cancellation is driven into a restricted rigidity branch unless another escape channel intervenes~\cite{LoboMunozTello2026}.  Thus passing a DK/scalaron mode-stability criterion does not by itself eliminate the independent Cauchy-horizon instability channel.  Conversely, a theory with a tachyonic scalaron introduces an additional instability that is conceptually distinct from inner-horizon blueshift.

The continued interest in mass inflation in higher-curvature regular black-hole models underscores that inner-horizon stability cannot be inferred from background regularity alone~\cite{FrolovZelnikov2026,DiFilippoEtAl2026}.  The present separation suggests a clean way to organize future analyses: first identify the background transverse transport, then establish scalaron hyperbolic/mode stability, and only then address the nonlinear Cauchy-horizon problem.

\subsection{\texorpdfstring{{The scalaron as a blueshifted source}}{The scalaron as a blueshifted source}}
\label{sec:scalarontail}

{The quadratic family makes the relation between Layers II and III concrete at the level of linear propagation on a fixed Reissner--Nordstr\"om background.  By Proposition~\ref{prop:constantF}, the scalaron obeys Eq.~\eqref{eq:RNscalarwave} to first order, in the interior as well as in the exterior.  It has an $\ell=0$ radiative degree of freedom, absent from spherical electrovacuum general relativity.  This mode can be excited by scalaron initial data or by traceful matter through Eq.~\eqref{eq:decoupled}, and can radiate in exactly spherical collapse.  Maxwell matter alone is tracefree, however, and does not force excitation of a scalaron initially at its constant background value.}

{For $\alpha>0$, let $m_{\rm sc}=1/\sqrt{6\alpha}>0$.  The massive Klein--Gordon field has an intermediate oscillatory tail $\propto t^{-(\ell+3/2)}\sin(m_{\rm sc}t)$ in the small-$m_{\rm sc}M$ regime~\cite{HodPiran1998}.  At asymptotically late times the exterior tail, under the scattering and initial-data conditions of Refs.~\cite{KoyamaTomimatsu2001a,KoyamaTomimatsu2001b,KoyamaTomimatsu2002}, takes the form}
\begin{equation}
{
\delta F\sim A\,t^{-5/6}\sin\!\big[m_{\rm sc}t+\vartheta(t)\big],}
\label{eq:massivetail}
\end{equation}
{independently of $\ell$, where $A$ is an amplitude and $\vartheta$ a slowly varying phase.  These massive tails decay more slowly than the corresponding massless Price tails for localized data~\cite{Price1972}.}

{To estimate interior focusing, we make two additional assumptions explicitly.  First, the event-horizon data have the same $v^{-5/6}$ oscillatory law in advanced time $v$.  Second, interior scattering transmits a nonzero coefficient of that law to the outgoing Cauchy-horizon branch under consideration, so that, at fixed outgoing coordinate $u$,}
\begin{equation}
{
\delta F(u,v)\sim A_-(u)\,v^{-5/6}
\sin\!\big[m_{\rm sc}v+\vartheta_-(u,v)\big],}
\label{eq:transmittedtail}
\end{equation}
{with $A_-(u)\neq0$ and $\partial_v\vartheta_-=o(m_{\rm sc})$.  These are hypotheses for the estimate, not a propagation theorem derived here.  The asymptotic law is understood with derivative control sufficient to differentiate its leading term.}

{Keep the signed convention of Eq.~\eqref{eq:Sxkappa}, so $\kappa_-<0$ and the blueshift rate is $|\kappa_-|=(r_+-r_-)/(2r_-^2)$.  Along outgoing null curves approaching the nondegenerate Cauchy horizon, affine parametrization gives $k\simeq C(u)e^{|\kappa_-|v}\partial_v$ and $d\lambda\simeq C(u)^{-1}e^{-|\kappa_-|v}dv$, where $C(u)>0$ fixes the affine normalization.  The explicitly quadratic term in Eq.~\eqref{eq:entropyfocus}, evaluated using the linear field and $F=1+\delta F$, therefore has the leading fixed-background estimate}
\begin{equation}
{
\big[k(\delta F)\big]^2\sim C^2 A_-^2m_{\rm sc}^2
e^{2|\kappa_-|v}v^{-5/3}\cos^2\!\Phi_-,}
\label{eq:scalaronsource}
\end{equation}
{where $\Phi_-=m_{\rm sc}v+\vartheta_-(u,v)$.  Its affine integral diverges: the squared oscillation is nonnegative and has a nonzero mean, while the remaining weight grows as $e^{|\kappa_-|v}v^{-5/3}$.  This is a conditional nonintegrability result for a contribution to generalized entropy focusing.  It uses a first-order field to estimate a second-order source; it is not a self-consistent second-order or nonlinear evolution.}

{The source must also be distinguished from the effective Jordan-frame longitudinal projection in the Hawking-mass transport law of Ref.~\cite{LoboMunozTello2026}, which contains the scalaron Hessian.  Applying that mass-inflation theorem requires checking its eventual source sign, longitudinal dominance, counterstreaming, and mixed-channel assumptions.  A positive squared kinetic term does not by itself verify them.  Massive-field interior results also distinguish derivative-energy blow-up from field-amplitude divergence and metric extendibility, with a decisive role for horizon oscillations~\cite{VanDeMoortel2018,KehleVanDeMoortel2024}.  Those Einstein--Maxwell--Klein--Gordon results guide the present discussion but are not nonlinear theorems for the full scalaron potential.  The robust conclusion here is that a locally non-tachyonic scalaron can supply an additional blueshift-amplified focusing contribution when it is excited and transmitted.  Its nonlinear consequences, including Hawking-mass inflation and the regularity class of a possible extension, remain questions for the coupled evolution.}

\subsection{Comparing instability channels}
\label{sec:channels}

{On the stationary asymptotically flat Reissner--Nordstr\"om background, a spectral growth rate $\Omega$ and the inner surface gravity $\kappa_-$ refer to the same Killing time, normalized at infinity.  Their ratio is invariant under a common constant rescaling of this time.  For a separated growing spectral component $\propto e^{\Omega t}$ that is regular at the future event horizon, the ingoing dependence there is $e^{\Omega v}$.  If its transmission coefficient to the ingoing Cauchy-horizon component is nonzero, the regular Kruskal coordinate $V_-=-e^{-|\kappa_-|v}$ gives the dependence $(-V_-)^{-\Omega/|\kappa_-|}$.  Accordingly,}
\begin{equation}
{
\mathfrak b\equiv\frac{\Omega}{|\kappa_-|}
\le\mathfrak b_{\max}
=\frac{2r_-^2}{(r_+-r_-)\sqrt{6|\alpha|}}.}
\label{eq:bratio}
\end{equation}
{This counts tachyonic $e$-folds per blueshift $e$-fold for that transmitted component.  In the continuous-spectrum setting of Sec.~\ref{sec:RN}, it is a componentwise estimate, not a claim of a discrete finite-energy growing eigenmode or a universal wave-packet asymptotic law.  It is analogous, with the opposite sign, to the decay-to-blueshift ratio $\beta$ used in Cauchy-horizon regularity analyses~\cite{Cardoso2018}: an exponentially decaying transmitted component behaves as $(-V_-)^\beta$, while a growing one behaves as $(-V_-)^{-\mathfrak b}$.  Under the stated transmission assumption the latter diverges at the level of the field.  For $\alpha>0$ in asymptotically flat spacetime, the late-time input instead has power-law decay and no positive exponential spectral gap, as in Sec.~\ref{sec:scalarontail}.}

{Neither ratio orders the nonlinear onset of the two mechanisms.  Perturbation amplitudes, propagation, and backreaction matter, and once the evolving geometry loses stationarity there is no common Killing clock.  A coupled characteristic evolution must compare the scalaron amplitude, blueshifted fluxes, and quasi-local mass with respect to the same double-null variables.}

\subsection{Strong-field viability filters and model-building implications}
\label{sec:implications}

The analysis suggests that viability in strong-field $f(R)$ gravity should be treated as a hierarchy of filters rather than a single sign condition.

{Table~\ref{tab:filters} summarizes the four complementary filters in the order used above.  Layer 0 requires a regular positive coupling, $F>0$.  Layer I identifies the longitudinal generalized entropy expansion and the transverse derivative type of entropy-marginal spheres.  Layer II analyzes the local coefficient $\Msq$ and the global coupled perturbation problem.  Layer III addresses nonlinear mass growth and Cauchy-horizon extendibility in a specified regularity class.  Passing any one filter does not establish the others; in particular, generalized entropy focusing is distinct from ordinary area focusing.}

Several consequences follow:
First, the traditional DK requirement $f_{RR}>0$ remains indispensable in the high-curvature GR-like regime, but the covariant local mass parameter~\eqref{eq:mscalar} is the more general local criterion.  A model may satisfy $f_{RR}>0$ and still have $\Msq<0$ on a sufficiently curved background.  Conversely, $F>0$ controls coupling and entropy regularity but contains no information about the second derivative of the Lagrangian.

Second, chameleon screening and strong-field stellar viability should not be conflated with null-transport conditions.  The construction of relativistic stars in viable $f(R)$ models relies on nonlinear scalar dynamics and the matter trace~\cite{UpadhyeHu2009,BabichevLanglois2010}; curvature singularities and their possible cure by higher-curvature terms represent yet another aspect of the scalaron potential~\cite{Frolov2008,KobayashiMaeda2009}.  These questions are naturally expressed in derivatives of $\U$, while entropy cross transport depends directly on $\U$ itself.

Third, dynamical environments can excite the scalaron even in the absence of a strictly tachyonic mass.  Contracting systems can drive high-frequency curvature oscillations~\cite{ArbuzovaDolgovReverberi2013}, and binary neutron-star simulations show that finite-range scalar degrees of freedom can alter inspiral and merger dynamics~\cite{SagunskiEtAl2018}.  A complete strong-field viability analysis should therefore combine the instantaneous mass criterion with adiabaticity and characteristic propagation.  Along a timelike observer $u^\mu$, a useful local measure is
\begin{equation}
\epsilon_{\rm ad}
\equiv
\frac{|u(m_{\rm sc})|}{m_{\rm sc}^2},
\label{eq:adiabatic}
\end{equation}
{when $m_{\rm sc}^2>0$, with $m_{\rm sc}=\sqrt{m_{\rm sc}^2}$ and $u(m_{\rm sc})=u^\mu\nabla_\mu m_{\rm sc}$.  Here $u^\mu$ is a unit timelike vector.  Small $\epsilon_{\rm ad}$ indicates slow variation of the local mass on its inverse-mass time scale; it does not alone establish adiabatic tracking of an evolving equilibrium.  Matter forcing, spatial gradients, metric evolution, and the physical mode frequency must also be controlled.  A value of order unity or larger flags possible nonadiabatic excitation and supplies a useful diagnostic for numerical relativity.}

Finally, the potential-jet hierarchy provides a compact model-building lesson.  A viable theory must control not only the value of $\U$, which affects algebraic curvature and entropy transport, but also its first and second derivatives, which govern the background curvature relation and scalaron response.  This is one reason why adding a positive high-curvature $R^2$ term can cure certain scalaron pathologies without trivially determining the causal structure of the resulting compact objects~\cite{Starobinsky2007,KobayashiMaeda2009}.

\section{Discussion and conclusions}
\label{sec:conclusions}

{The original Dolgov--Kawasaki calculation asked whether a nominally GR-like matter configuration remains close to its general-relativistic curvature.  Its significance was to show that a modification intended for cosmological curvature can nevertheless cause a rapid local instability in ordinary matter~\cite{DolgovKawasaki2003}.  The covariant formulation developed here retains that physical question while separating the scalaron response from the geometric transport of a chosen background.  Equation~\eqref{eq:lineartrace} is exact at first order on an arbitrary background: $\Msq$ is the local restoring coefficient, while varying $F$ introduces metric mixing and matter-trace perturbations provide forcing.  On constant-$F$ backgrounds the scalaron equation decouples in one direction, and its spherical reduction~\eqref{eq:characteristic} supplies a characteristic propagation equation.}

{The background entropy identities~\eqref{eq:entropyfocus} and~\eqref{eq:crossentropy} cancel the scalaron Hessian.  The Einstein-frame dictionary explains the distinction: longitudinal focusing samples scalar kinetic flux, while spherical cross-focusing samples the potential value; neither directly measures its curvature in field space.  This does not imply dynamical independence of scalaron and metric perturbations.  It means that these particular background transport data cannot determine the scalaron mass sign over theory space.  The cosmological check in Appendix~\ref{app:FLRW} confirms the same distinction on a background with time-dependent $F$.}

{The sharpest discriminator is quadratic electrovacuum gravity.  The same nonextremal Reissner--Nordstr\"om spacetime, with the same complete background transport and outer--inner horizon pair, solves $f(R)=R+\alpha R^2$ for every $\alpha$.  For $\alpha\neq0$, its scalaron mass is $1/(6\alpha)$: the positive branch has a nonnegative exterior scalaron potential, whereas the negative branch has a tachyonic spectral band.  Theorem~\ref{thm:onejet} identifies the underlying freedom on all constant-curvature, tracefree-matter backgrounds: the embedding and background transport are fixed by $(f,f_R)$, but $f_{RR}$ can vary independently.  Rotating examples further show that positive mass does not exclude superradiant mode instability in suitable regimes.}

{The running-scalaron benchmark adds a complementary lesson.  Within a particular theory, the background equations can correlate the mass coefficient and horizon transport, as in Eq.~\eqref{eq:TWPcorrelation}.  Such a relation does not contradict separation over theory space.  It also reveals why entropy marginality must be distinguished from ordinary area marginality: the sphere $r=2\mu$ has two vanishing generalized expansions without being a Killing horizon.  Since the benchmark has a singular centre and loses positive coupling at finite radius, it tests local identities rather than providing a globally viable compact-object model.}

{The scalaron can also connect the stability layers through propagation.  Assuming an excited massive tail with nonzero interior transmission and the stated derivative asymptotics, Eq.~\eqref{eq:scalaronsource} gives a divergent affine integral of a quadratic contribution to generalized entropy focusing on the fixed background.  This is a conditional estimate, not a nonlinear mass-inflation theorem.  Identifying it with the effective source of the companion Hawking-mass transport law requires checking that law's sign, dominance, and counterstreaming hypotheses~\cite{LoboMunozTello2026}.  Field amplitude, derivative energy, Hawking mass, and metric extendibility must be tracked separately, especially for oscillatory massive tails~\cite{VanDeMoortel2018,KehleVanDeMoortel2024}.  Likewise, the stationary rate ratio~\eqref{eq:bratio} compares transmitted spectral components, not the nonlinear onset of two competing mechanisms.}

{A decisive next step is a fully coupled double-null evolution of a model whose scalaron coefficient is positive outside the black hole but becomes negative only in a high-curvature interior.  The present running benchmark realizes a spatial sign change with the negative region outside the horizons, so it does not yet supply that model.  For varying $F$, the characteristic scalar equation must retain the metric-mixing term of Eq.~\eqref{eq:lineartrace}; the full nonlinear problem additionally requires the gravitational constraints, remaining evolution equations, and matter dynamics.  Monitoring $\Sx$, $\Theta_{(\ell,n)}$, $\Msq$, both physical and effective null fluxes, and the Hawking/Misner--Sharp mass would determine whether scalaron growth enhances, delays, or suppresses a mass-inflation regime.}

{Relativistic stellar collapse and neutron-star mergers provide another application.  For an isotropic perfect fluid the trace $T=-\rho+3p$, with energy density $\rho$ and pressure $p$, can change rapidly across shocks and phase transitions.  Tracking the local mass, matter forcing, and the qualified adiabaticity diagnostic~\eqref{eq:adiabatic} would extend the original matter-instability test to these environments; existing finite-range scalar simulations establish the feasibility of related calculations~\cite{SagunskiEtAl2018}.  A further extension to general scalar--tensor or higher-curvature theories must derive transport and perturbation equations from their own kinetic and potential structure, rather than assume $\U_{,FF}=1/f_{RR}$.  The layered framework in Table~\ref{tab:filters} organizes these questions while keeping local response, background transport, and nonlinear inner-horizon stability distinct.}

\begin{acknowledgments}
FSNL acknowledges support from the Funda\c{c}\~ao para a Ci\^encia e a Tecnologia (FCT) research grant UID/04434/2025.
\end{acknowledgments}

\appendix
\section{\texorpdfstring{{Cosmological check of the transport identities}}{Cosmological check of the transport identities}}
\label{app:FLRW}

{As an independent check on a background with time-dependent $F$, consider spatially flat FLRW,}
\begin{equation}
{
ds^2=-dt^2+a^2(t)\left(d\chi^2+\chi^2d\Omega^2\right),}
\label{eq:FLRW}
\end{equation}
{with areal radius $r=a\chi$, $H=\dot a/a$, and a perfect fluid $(\rho,p)$ comoving with $u=\partial_t$.  The gauge~\eqref{eq:crossgauge} is realized by}
\begin{equation}
{
n=\frac1a\left(\partial_t-\frac1a\partial_\chi\right),
\qquad
\ell=\frac a2\left(\partial_t+\frac1a\partial_\chi\right),}
\label{eq:FLRWnull}
\end{equation}
{and $k=a^{-1}(\partial_t+a^{-1}\partial_\chi)$ is an affinely parametrized outgoing null vector.  Then $\theta_{(k)}=(2/a)(H+1/r)$, $k(F)=\dot F/a$ and $T_{kk}=(\rho+p)/a^2$, and Eq.~\eqref{eq:entropyfocus} reduces to}
\begin{equation}
{
-2F\dot H=8\pi(\rho+p)+\ddot F-H\dot F.}
\label{eq:FLRWray}
\end{equation}
{For the cross transport, $\Jx=\tfrac12ar^2(\dot F+2HF)+arF$ and $T_{\ell n}=\tfrac12(\rho-p)$, and Eq.~\eqref{eq:crossentropy} reduces to}
\begin{equation}
{
\ddot F+5H\dot F+2F\big(\dot H+3H^2\big)=8\pi(\rho-p)+\U(F),}
\label{eq:FLRWcross}
\end{equation}
{which is Eq.~\eqref{eq:FLRWray} plus twice the Friedmann constraint $3FH^2=8\pi\rho+\tfrac12\U-3H\dot F$.  Both are the standard cosmological equations of metric $f(R)$ gravity~\cite{SotiriouFaraoni2010,DeFeliceTsujikawa2010}.  Together with $R=6(2H^2+\dot H)$ they imply the background trace equation~\eqref{eq:traceU}, as they must, but neither contains $\U_{,FF}$.  The scalaron response appears only upon perturbing.  At an expanding de Sitter fixed point, $H>0$ and $F$ is constant.  Proposition~\ref{prop:constantF} applies, and a homogeneous perturbation obeys}
\begin{equation}
{
\delta\ddot F+3H\,\delta\dot F+\Msq\,\delta F=-\frac{8\pi}{3}\,\delta T,}
\label{eq:FLRWresponse}
\end{equation}
{For the unforced equation, $\delta T=0$, the characteristic exponents satisfy $s^2+3Hs+\Msq=0$.  Thus $\Msq>0$, equivalently $F/f_{RR}>R$, gives decay of both homogeneous modes, whereas $\Msq<0$ gives a growing mode.  At $\Msq=0$ there is a constant mode and a decaying mode: the point is marginal in this linear test, not an asymptotic attractor.  Accordingly, $\Msq\ge0$ excludes exponential homogeneous growth, while strict positivity gives linear asymptotic attraction in this scalar sector~\cite{AmendolaEtAl2007,SotiriouFaraoni2010,DeFeliceTsujikawa2010}.  Matter dynamics must still be included if $\delta T$ is not prescribed independently.  The cosmological background thus displays the same division of labor as the black-hole examples: Eqs.~\eqref{eq:entropyfocus} and~\eqref{eq:crossentropy} are background equations, whereas Eq.~\eqref{eq:lineartrace} is the response.}


\end{document}